\documentclass{article}
\usepackage{amsmath, amssymb, amsthm}
\usepackage{fullpage}
\usepackage{verbatim}
\usepackage{microtype}
\usepackage{graphicx}
\usepackage[margin=1.1in]{geometry}
\usepackage{times}
\usepackage[utf8]{inputenc}
\title{Some Pairs Problems}
\author{Jeffrey D.~Ullman\thanks{Stanford University} \and
Jonathan R.~Ullman\thanks{Northeastern University}}
\date{\today}
\newcommand{\pr}[2]{\underset{#1}{\mathbb{P}}\left[ #2 \right]}
\newcommand{\ex}[2]{\underset{#1}{\mathbb{E}}\left[ #2 \right]}

\newcommand{\set}[1]{\left\{#1\right\}}

\newtheorem{theorem}{Theorem}[section]

\newtheorem{lemma}[theorem]{Lemma}

\newtheorem{claim}[theorem]{Claim}

\newtheorem{corollary}[theorem]{Corollary}

\newtheorem{example}[theorem]{Example}
\theoremstyle{definition}

\newtheorem{definition}[theorem]{Definition}
\begin{document}
\maketitle
\begin{abstract}
A common form of MapReduce application involves discovering relationships between certain pairs of inputs. Similarity joins serve as a good example of this type of problem, which we call a ``some-pairs'' problem. In the framework of \cite{ADSU}, algorithms are measured by the tradeoff between reducer size (maximum number of inputs a reducer can handle) and the replication rate (average number of reducers to which an input must be sent. There are two obvious approaches to solving some-pairs problems in general. We show that no general-purpose MapReduce algorithm can beat both of these two algorithms in the worst case. We then explore a recursive algorithm for solving some-pairs problems and heuristics for beating the lower bound on common instances of the some-pairs class of problems.
\end{abstract}

\section{Introduction}
\label{intro-sect}
In \cite{ADSU}, MapReduce \cite{deanGhemawat} algorithms were studied from the point of view of finding the tradeoff between the {\em reducer size} (maximum number of inputs that can be sent to any reducer) and the {\em replication rate} (average number of reducers to which an input is sent). In this model, a {\em problem} is a mapping between a set of inputs and a set of outputs. A MapReduce algorithm that solves a problem is called a {\em mapping schema}; it is an assignment of inputs to reducers so that no reducer gets more inputs than the reducer size allows, yet for every output there is at least one reducer that gets all the inputs associated with that output.

\subsection{Some-Pairs Problems}
\label{some-pairs-subsect}
A particular result of this type covered in \cite{Ull-XRDS} shows that for the {\em all-pairs} problem (one output is associated with each pair of inputs), the replication rate is at least the number of inputs divided by the reducer size. Moreover, this bound is essentially tight. Here, we are going to look at the family of problems where all outputs are associated with exactly two inputs, but not all pairs of inputs are associated with an output. This model was first studied for MapReduce by \cite{OR11} as the problem of theta-joins.

It turns out to be convenient to view this family of problems as "X-Y" problems in the sense of \cite{ADKSU}. Here, there are two sets of $n$ inputs each, called $X$ and $Y$, and each output is a pair of inputs, one chosen from $X$ and the other from $Y$. We call this family {\em Some-Pairs Problems}.

\begin{figure}[htb]
\centerline{\includegraphics[width=0.25\textwidth]{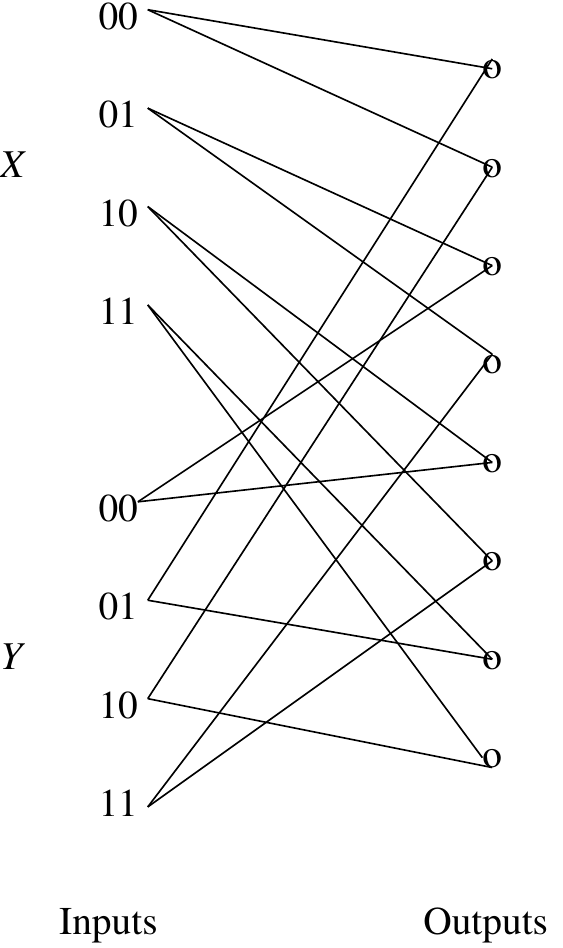}}
\caption{The Hamming-distance-1 problem as an input/output mapping}
\label{XYproblem-fig}
\end{figure}

\begin{example}
\label{hd-ex}
There are many examples of problems of the Some-Pairs type. A rich source of such problems is similarity joins or ``fuzzy'' joins, as discussed in \cite{FuzzyJoin} and \cite{AnchorPoints}. Perhaps the simplest of these problems is that of finding bit strings at Haming-distance 1. That is, the inputs are bit strings of some fixed length, and there is an output for each pair, one from $X$ and one from $Y$ that differ in exactly one bit.\footnote{One might imagine that if all these inputs are present in the input file(s), then the answer needs no calculation. Rather we know what pairs are at Hamming-distance 1 and can simply write them out with no calculation.

However, in the model of \cite{ADSU}, inputs are ``hypothetical,'' and not all may be present at any one time. If, say, only 10\% of the inputs are expected to be present, then the reducer size can be scaled up by a factor of 10, so that the expected number of actual inputs at any reducer matches the limt on that quantity.} Figure~\ref{XYproblem-fig} shows the representation of this problem for the very simple case of bit strings of length two. there are four strings in each of the sets $X$ and $Y$. Each member of $X$ is Hamming-distance 1 from two of the strings in $Y$, so there are eight outputs in all. For example, the topmost output in the figure represents 00 from $X$ and 01 from $Y$, which differ in their second bits.
\end{example}

\subsection{Some-Pairs Problems as Bipartite Graphs}
\label{some-pairs-bipartite-subsect}
We shall find it useful to represent some-pairs problems by a bipartite graph, which we call the {\em connection graph} for the problem. In these graphs, the two sets of nodes correspond to the input sets $X$ and $Y$. There is an edge between two nodes if and only if there is an output associated with that pair of nodes. For example, the Hamming-distance-1 problem from Example~\ref{hd-ex} can be represented by the graph in Fig.~\ref{XYgraph-fig}.

\begin{figure}[htb]
\centerline{\includegraphics[width=0.20\textwidth]{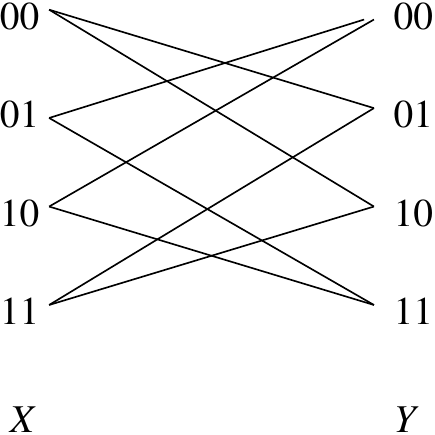}}
\caption{The Hamming-distance-1 problem as a bipartite ``connection'' graph}
\label{XYgraph-fig}
\end{figure}

\section{Upper Bounds on a Some-Pairs Problem}
\label{upper-sect}
The first observation is that we can treat a Some-Pairs problem as if it were an X-Y problem. In this paper, we are going to suppose that a reducer can hold $q$ inputs from the set $X$ and also hold $q$ inputs from set $Y$. Thus, it is possible for a reducer to handle $2q$ inputs, as long as these inputs are divided equally between the two input sets.  We say such a reducer is of {\em size} $q$.\footnote{In all other works using this model, the ``reducer size'' $q$ was used as the limit on the total number of inputs allowed at a reducer.}  In this manner, one reducer can hold $q$ edges of the connection graph. Further, we shall always use $n$ for the sizes of both sets of inputs $X$ and $Y$, and we shall use $m$ for the number of outputs the problem has (or equivalently, the number of edges in its connection graph).

There are two very simple algorithms that can handle the case where $m$ is significantly less than $n^2$. Moreover, these algorithms, together, are close to the best that can be done in general, as we shall show in Section~\ref{lower-sect}. These are:
Algorithm A: Ignore the fact that $m<n^2$ and use a simple algorithm that works when there is an output for each pair of inputs, one from $X$ and one from $Y$.
Algorithm B: Create one reducer (whose size can be $q=1$) for each edge.

\subsection{Algorithm A}
Following the technique of \cite{Ull-XRDS}, we can partition each of the sets $X$ and $Y$ into $n/q$ groups. There will be one reducer for each pair of groups, one from $X$ and one from $Y$. The replication rate is then $r = n/q$, since each member $x$ of $X$ has to be sent to the $n/q$ reducers that correspond to $x$'s group and one of the $n/q$ groups from $Y$. A similar observation holds for the members of $Y$, so the average replication, and in fact the exact replication for each input, is $n/q$. Thus, $n/q$ is surely an upper bound on the replication rate for any Some-Pairs problem, regardless of how large or small $m$ is.

\subsection{Algorithm B}
Create one reducer for each output. Send each input to the reducers for the outputs in which that input participates. Thus, the number of reducers to which any input is sent is its degree. The sum of the degrees of the inputs in $X$ is $m$, so the average degree is $m/n$. The same holds for the inputs in $Y$. Thus, the replication rate for Algorithm B is $m/n$.

Depending on the relative sizes of $n$, $m$, and $q$, either of Algorithms A and B could be better. Algorithm B is better when $n/q > m/n$, or $m<n^2/q$. That is, Algorithm B is preferred when $m$ is small compared with its maximum possible value of $n^2$. But how small is ``small'' depends on $q$. If $q$ is large, then $m$ has to be very small indeed, or else it is better to assume all edges (outputs) are part of the problem and use Algorithm A.

\section{A Lower Bound on General-Pur\-pose Algorithms}
\label{lower-sect}
It should be obvious that Algorithms A and B are ``general-purpose,'' in the sense that they make no use of the structure of the problem itself. There are numerous examples of problems where one can take advantage of a particular structure, such as the Hamming-distance-1 problem discussed in Example~\ref{hd-ex}, where one can construct a mapping schema with a much lower replication rate than that given by either algorithm.

However, we shall show in this section that essentially no improvements can be made in general. That is, for any $n$, $m$, and $q$ there must exist problems for which there are no algorithms that are significantly better than the better of Algorithms A and B.

\subsection{Complete Reducers}
\label{complete-subsect}
Our first step is to simplify the problem by showing we can restrict our thinking to mapping schemas where every reducer gets the maximum allowable number of inputs, $q$.  In what follows, we shall continue to use $n$ for the number of members of input sets $X$ and $Y$, $m$ for the number of outputs (or equivalently, edges in the bipartite graph between sets of nodes $X$ and $Y$), and $q$ for the number of inputs from $X$ and the number of inputs from $Y$ that a reducer can receive. That is, a reducer can receive up to $2q$ inputs, but we assume that at most $q$ come from one of the sets.
We make two assumptions of nontriviality:

\begin{enumerate}
\item
We shall assume that $q<n/2$. The purpose of this assumption is to exclude problems with solutions that require only one or two reducers.
\item
We also assume that each input is needed for at least one output. This assumption is a little tricky. We want to assume the two input sets $X$ and $Y$ have the same size $n$. But if some of the inputs from one set are not needed for any output, then we would have to delete them, making the sizes of the two sets unequal. However, we could add at most $n$ ``dummy'' outputs to the problem, which would add at most 1 to the replication rate, and would guarantee that each input was needed.
\end{enumerate}

An algorithm to solve the problem will, as normal, be represented by a mapping schema that assigns each input to a set of reducers. We shall use $p$ for the number of reducers used by an algorithm, and $r$ for the replication rate, or average number of reducers to which a given input is sent.

We say an algorithm or mapping schema is {\em complete} if every reducer is assigned exactly $q$ inputs from each of the two input sets $X$ and $Y$. Note that for complete mapping schemas, the equation $rn = pq$ holds. The first step in our lower-bound proof is to show that, to within a constant factor, the replication rate of the best complete mapping schema for a given problem is the same as that for any mapping schema for the same problem. As a result, we can restrict our analysis to complete mapping schemas.

\begin{lemma}
\label{complete-lemma}
For any some-pairs problem $P$ with $n$ inputs in each of $X$ and $Y$ and with $m$ outputs, and any reducer size $q<n/2$, the minimum replication rate for a complete mapping schema for $P$ is no more than 6 times the minimum replication rate for any mapping schema for $P$.
\end{lemma}

\begin{proof}
Consider some mapping schema $M_1$ that solves $P$. If $M_1$ has two reducers that each get no more than $q/2$ inputs from $X$ and no more than $q/2$ inputs from $Y$, then we can combine these two reducers into one. If we repeat this modification until no more changes are possible, we are left with a new mapping schema $M_2$ with the same replication rate (or less if we happen to combine two reducers whose input sets were not disjoint). $M_2$ has the property that at most one of its reducers has fewer than $q/2$ inputs.

Moreover, since $q<n/2$, and we assume every input is needed, there are at least 3 reducers. If there are $p\ge3$ reducers, then one might be almost empty, while the other $p-1$ each have at least $q/2$ of the possible $2q$ inputs. We can add inputs to any reducers that do not have $q$ inputs from $X$ and $q$ inputs from $Y$, and thus create a complete mapping schema $M_3$ that also solves the problem $P$ and has at most six times the replication rate of $M_2$ and $M_1$.
\end{proof}

\subsection{The Lower Bound on Replication Rate}
\label{lower-bound-subsect}
In the previous section we established that for every some-pairs problem, the replication rate is $r \lesssim \min\{m/n, n/q\}$.  In this section we shall prove a nearly match worst-case lower bound.  Specifically, we will prove that for every $n, m$, there exists a some-pairs problem that requires replication rate
$$
r \gtrsim \min\left\{ \frac{m}{n}, \frac{n}{q} \right\}.
$$
Let $X$ denote the set of inputs and $E$ denote the set of edges. The proof will rely on some-pairs problems (equivalently, graphs) that satisfy a quantitatively strong \emph{edge-isoperimetric inequality}. Roughly, an edge-isoperimetric inequality upper bounds the number of edges that are \emph{covered} by any subset of a given size.\footnote{Typically edge-isoperimetric inequalities are stated as lower bounds on the number of edges that ``leave'' any subset of a given size. For regular graphs (those in which edge node has the same number of neighbors) these two forms of the inequality are equivalent.}  Given two sets of inputs $S \subseteq X$ and $T\subseteq Y$, the set of \emph{edges covered by $S$ and $T$} is
$$
C(S,T) = \set{ (u,v) \in E \; | \; u \in S, v \in T }
$$

\begin{comment}We assume that the graph is \emph{regular}, which here means that $m/n$ is an integer and each input is incident to exactly $m/n$ edges. We will use $d = m/n$ to denote the degree of the graph.
\end{comment}

\begin{definition}[Expansion]
We say that a some-pairs problem is \emph{$(q, \phi)$-expanding} if for every sets of inputs $S \subseteq X$ and $T\subseteq Y$, each of size at most $q$, $|C(S,T)| \leq \phi$.
\end{definition}

Note that generally, $\phi$ will be a function of $q$.    This definition is very close to the standard definition of \emph{edge expansion} with two differences.  The first, and most substantive, difference is that we only consider the expansion of sets $S,T$ of size at most $q$.  The second is purely notational.  We count the number of edges \emph{in} $C(S,T)$ whereas typically expansion is defined as the number of edges with exactly one endpoint in $S$ or $T$.  We choose this notation to emphasize the fact that we are interested in graphs with ``nearly perfect expansion,'' meaning very few edges are covered.  See the survey of Hoory, Linial, and Wigderson~\cite{HooryLW06} for a textbook reference on expander graphs.

The following lemma shows that if a some-pairs problem is expanding, then it requires a large replication rate.

\begin{lemma}
\label{expand-lemma}
If $(X,E)$ is a some-pairs problem that is $(q, \phi)$-expanding, then any complete mapping schema using reducers of size $q$ must use at least $p \geq m/\phi$ reducers, and thus has a replication rate of $r \geq qm/\phi n$.
\end{lemma}

\begin{proof}
In any complete mapping schema, for every edge $(u,v) \in E$, there must exist some reducer that is assigned both endpoints $u$ and $v$. That is, some reducer must cover the edge $(u,v)$. But if this some-pairs problem is $(q, \phi)$-expanding then any reducer of size at most $q$ can cover at most $\phi$ such edges. Since there are $m$ edges in total that must be covered, there must be at least $p \geq m/\phi$ reducers. The lower bound on the replication rate follows from the identity $r = pq/n$.
\end{proof}

\begin{example}
Consider the some-pairs problem of finding all strings of Hamming distance $1$. Here $X = \{0,1\}^{\log_2(n)}$ and $E$ consists of all pairs of strings with Hamming distance $1$. Note $|E| = n \log_2(n)$. One can show that for any set of size $q,$ the number of covered edges is at most $q\log_2(2q)$ (see \cite{ADSU}, e.g., but note that $q$ here is $2q$ there). Thus, this some-pairs problem is $(q, q \log_2(2q))$-expanding for every $q$.  By Lemma~\ref{expand-lemma}, this problem requires replication rate
$$r \geq qm/\phi n = \log_2(n) / \log_2(2q)$$
\end{example}

Now our goal is to prove a worst-case lower bound on the replication rate, as a function of $m$, $n$, and $q$ by considering some-pairs problems with the best possible expansion (the lowest possible values of $\phi$).  For all standard definitions of expansion, random graphs are known to achieve effectively the best possible parameters (see e.g.~\cite{Pinsker73, Bassalygo81} for classical results in the area).  Thus, for our lower bound we will consider a random some-pairs problem with $n$ inputs on each side and $m$ edges, and analyze its expansion under our definition.  A useful tool for analyzing random graphs is the \emph{Chernoff Bound}, which we state in simplified form below.

\begin{lemma}[Chernoff bound]
Let $A_1,A_2,\dots,A_m$ be a set of independent random variables that take values in $\{0,1\}$. Let $A = \sum_{j=1}^{m} A_{j}$ and let $\mu = \ex{}{A} = \sum_{j=1}^{m} \ex{}{A_{j}}$. Then we have the following bounds on the probability that $A$ takes large values.

\begin{enumerate}
\item For $\delta \geq 2e - 1$,
$
\pr{}{ A > (1+\delta) \mu } \leq 2^{-\delta \mu}
$
\item For $\delta \leq 2e - 1$,
$
\pr{}{ A > (1+\delta) \mu } \leq e^{-\delta^2 \mu / 4}
$
\end{enumerate}
\end{lemma}

We shall use the Chernoff bound to show the following.

\begin{claim}
Let $X$ be a set of $n$ inputs and $E$ be a set of $m$ edges chosen independently at random. Let $S$ be a subset of $X$ of size at most $q$ and $T$ be a subset of $Y$ of size at most $q$. Then

\begin{enumerate}
\item For $\delta \geq 2e - 1$
$
\pr{}{ |C(S,T)| > (1+\delta)(mq^2/n^2) } \leq 2^{-\delta \mu}
$
\item For $\delta \leq 2e - 1$
$
\pr{}{ |C(S,T)| > (1+\delta)(mq^2/n^2) } \leq e^{-\delta^2 \mu / 4}
$
\end{enumerate}
\end{claim}

\begin{proof}
Let $A_{j}$ be a random variable that takes the value $1$ if the $j$-th edge is covered by $(S,T)$ and the value $0$ otherwise. Since the edges are chosen independently, $A_1,\dots,A_m$ are independent. Since $|S|, |T| \leq q$, $\ex{}{A_j} \leq q^2/n^2$. Moreover, $|C(S,T)| = A = \sum_{j=1}^{m} A_{j}$ and $\mu = \ex{}{A} = \sum_{j=1}^{m} \ex{}{A_j} \leq m(q^2 /n^2)$. The claim is now immediate from the Chernoff bound.
\end{proof}

The next step is to put a bound on the probability that \emph{any} pair of sets $(S,T)$ of size at most $q$ covers a large number of edges. The next claims follows from the fact that there are at most $\binom{n}{q}^2 \leq (en/q)^{2q}$ such sets.

\begin{claim}
Let $X$ be a set of $n$ inputs and $E$ be a set of $m$ edges chosen independently at random. Then

\begin{enumerate}
\item For $\delta \geq 2e - 1$
$$
\pr{}{ \max_{S \subseteq X, T \subseteq Y\; |S|,|T| = q} |C(S,T)| > (1+\delta)(mq^2/n^2) } \leq 2^{2q \log_2(en/q) - \delta \mu}
$$
\item For $\delta \leq 2e - 1$
$$
\pr{}{\max_{S \subseteq X, T \subseteq Y\; |S|,|T| = q} |C(S,T)| > (1+\delta)(mq^2/n^2) } \leq e^{2q \log_e(en/q) - \delta^2 \mu / 4}
$$
\end{enumerate}
\end{claim}

Now we can prove the lower bound on replication rate via case analysis.

\subsubsection*{Case 1: $q \leq n^{2}/m$.}
Recall $\mu = \ex{}{|C(S,T)|} = mq^2 / n^2$.
Set $\delta$ so that
$$
\delta = \frac{3 q \log_2(en/q)}{\mu} = \frac{3 n^2 \log_2(en/q)}{mq}
$$
Note that, since $q \leq n^{2}/m$, we have $\delta \geq 2e-1$ as long as $m \geq \kappa n$ for some absolute constant $\kappa > 0$ ($\kappa = 2^{3e}/e$ would suffice). Thus, we can apply the Chernoff bound to obtain
\begin{align*}
&\pr{}{ \max_{S \subseteq X, T \subseteq Y\; |S|,|T| = q} |C(S,T)| > 2\delta \mu} \leq 2^{2q \log_2(en/q) - \delta \mu} \\
\Longrightarrow{} &\pr{}{ \max_{S \subseteq X, T \subseteq Y\; |S|,|T| = q} |C(S,T)| > 6 q \log_2(en/q)} \leq 2^{2q \log_2(en/q) - 3q \log_2(en/q)} = 2^{- q \log_2(en/q)} < 1.
\end{align*}
Thus there exists a some-pairs problem that is $(q, \phi)$-expanding for $\phi \leq 4q \log_2(en/q)$. Using our lemma, in this case we have a replication rate of
$$
r \geq \frac{q m}{\phi n} \geq \frac{m}{n} \cdot \frac{1}{6 \log_2(en/q)}
$$
which nearly matches the upper bound of $r \leq m/n$.

\subsubsection*{Case 2: $q \geq n^{2} \log_e(en/q) /m$.}
Set $\delta$ so that
$$
\delta^2 = \frac{3 q \log_e(en/q)}{\mu} = \frac{3 n^2 \log_e(en/q)}{mq}
$$
In this case, since $q \geq n^{2} \log_e(en/q)/m$, we have $\delta \leq \sqrt{3} \leq 2e - 1$. Thus, we can apply the second form of the Chernoff bound to obtain
$$
\pr{}{ \max_{S \subseteq X, T \subseteq Y\; |S|,|T| = q} |C(S,T)| > (1+\delta) \mu} \leq e^{2q \log_e(en/q) - \delta^2 \mu}
$$
First, by our choice of $\delta$, we have $1+\delta \leq 3$. Also recall that $\mu = mq^2/n^2$. Substituting for $\delta$ and $\mu$ and simplifying, we have
$$
\pr{}{ \max_{S \subseteq X, T \subseteq Y\; |S|,|T| = q} |C(S,T)| > \frac{ 3 m q^2 }{ n^2 }} \leq e^{2q \log_e(en/q) - \delta^2 \mu}
$$
By our choice of $\delta^2$, we have
$$
\pr{}{ \max_{S \subseteq X, T \subseteq Y\; |S|,|T| = q} |C(S, T)| > \frac{ 3 m q^2 }{ n^2 }} \leq e^{- q \log_e(en/q)} < 1.
$$
Thus there exists a some-pairs problem that is $(q, \phi)$-expanding for $\phi \leq 3m q^2 / n^2$.  It follows that
$$
r \geq \frac{q m}{\phi n} \geq \frac{q m}{ (3m q^2 / n^2) n} = \frac{n}{3q}
$$
which nearly matches the upper bound of $r \leq n/q$.

When we put the two cases together, we have the following lower bound on replication rate:

\begin{theorem}
\label{lower-th}
For any mapping schema that uses complete reducers,
$$
r\ge \min\left(\frac{n}{3q}, \frac{m}{n} \cdot \frac{1}{6 \log_2(en/q)}\right)
$$
\end{theorem}

Further, when we put Theorem~\ref{lower-th} together with Lemma~\ref{complete-lemma} we have a similar lower bound for all mapping schemas.

\begin{corollary}
\label{lower-corr}
For any mapping schema whatsoever,
$$
r\ge \min\left(\frac{n}{18q}, \frac{m}{n} \cdot \frac{1}{6 \log_2(en/6q)}\right)
$$
\end{corollary}

\section{A General Algorithm}
\label{alg-c-sect}
While we cannot beat the better of Algorithms A and B in all cases, there is a recursive approach to decomposing some-pairs problems that is close to the better of these algorithms, and offers the opportunity in many real examples to do better than either of these algorithms.

\subsection{Algorithm C}
\label{alg-c-subsect}
Suppose we are given a problem with $n$ inputs (nodes) and $m$ outputs (edges) represented as before by a bipartite graph with two sets of nodes $X$ and $Y$. Let $q$ be the reducer size.

\smallskip

\noindent
{\small BASIS}: If $n\le q$ use one reducer that gets all the nodes from $X$ and $Y$. Or, if $m\le q$, use one reducer that receives the nodes from $X$ and $Y$ that are endpoints of one of the $m$ edges. Note there can be no more than $q$ of either.

\smallskip

\noindent
{\small INDUCTION}: Divide the inputs of $X$ into two equal-sized groups $X_1$ and $X_2$, of $n/2$ inputs each, and do the same for $Y$, dividing it into groups $Y_1$ and $Y_2$. We thus have four subproblems, each with two input sets of size $n/2$: $(X_1,Y_1)$, $(X_1,Y_2)$, $(X_2,Y_1)$, and $( X_2,Y_2)$. Solve each of these recursively.

\subsection{Analysis of Algorithm C}
\label{alg-c-anal-subsect}
We can prove that the replication rate for Algorithm C is at most $\sqrt{m/q}$. The proof is an induction on $n$, To that end, define $r_q(n,m)$ to be the maximum replication rate needed to solve any Some-Pairs problem with $n$ inputs in each set, $m$ outputs, and reducer size $q$, using Algorithm~C.

\begin{theorem}
\label{alg-c-th}
$r_q(n,m) \le \sqrt{m/q}$ for any $n = q2^i$, where $i\ge0$.
\end{theorem}

\begin{proof}
The proof is an induction on $n$.

\smallskip

\noindent
{\small BASIS}:
For the basis, let $n=q$. If $m\ge q$, then $\sqrt{m/q} \ge 1$. But since $n=q$ we can send all nodes to one reducer, giving a replication rate of 1, so the theorem holds.

If $m<q$, then we can send only the at most $m$ nodes from each of the two sets $X$ and $Y$ that participate in edges to one reducer. The replication rate is thus at most $m/q$. But since $m<q$, $m/q<\sqrt{m/q}$, again proving the theorem.  Note that in this case, the replication rate is actually a proper fraction, but that is fine, since most inputs do not participate in an edge, and therefore are never communicated to any reducer; i.e., their replication is 0.

\smallskip

\noindent
{\small INDUCTION}:
Suppose $n=q2^i$ for some $i>0$. Divide the set of nodes $X$ and $Y$ into two equal-sized sets, each of size $q2^{i-1}$, to make the four subproblems $(X_1,Y_2)$, $(X_2,Y_1)$, and $( X_2,Y_2)$ as described in Algorithm~C. The $m$ edges of the original connection graph will divide among those four subproblems in an unknown way. Suppose that $m_{ij}$ is the number of edges that belong to the subproblem $(X_i,Y_j)$. We know that $$m_{11}+m_{12}+m_{21}+m_{22} = m$$
By the inductive hypothesis, we know that $r_q(n/2,m_{ij}) \le \sqrt{m_{ij}/q}$ for $i,j \in \{1,2\}$.

Think first of the inputs in $X_1$. These must be replicated an average of $\sqrt{m_{11}/q}$ times for the subproblem that pairs $X_1$ with $Y_1$, and they must be paired an average of $\sqrt{m_{12}/q}$ times for the subproblem that pairs $X_1$ with $Y_2$.  Their average replication rate is thus at most $\sqrt{m_{11}/q} + \sqrt{m_{12}/q}$. Similarly, we can argue that the average replication rate for the inputs in $X_2$ is at most $\sqrt{m_{21}/q} + \sqrt{m_{22}/q}$. It therefore follows that the average replication rate for inputs in $X$ is at most the average of these two bounds, or $$\frac12\bigl(\sqrt{m_{11}/q}+\sqrt{m_{12}/q}+\sqrt{m_{21}/q}+\sqrt{m_{22}/q}\bigr)$$
The same analysis applies to the members of $Y$, so we can assert that the replication rate needed by this recursive approach to constructing a mapping schema is at most
$$r_q(n,m) \le \frac1{2\sqrt{q}}\bigl(\sqrt{m_{11}}+\sqrt{m_{12}}+\sqrt{m_{21}}+\sqrt{m_{22}}\bigr)$$

We next need to show that
$$\sqrt{m_{11}}+\sqrt{m_{12}}+\sqrt{m_{21}}+\sqrt{m_{22}} \le 2\sqrt{m}$$
Since the square-root function is concave, we know that the maximum occurs when all the $m_{ij}$'s are equal, that is, they are each $m/4$. However, in more detail, we can set up the Lagrangian with the constraint $m_{11}+m_{12}+m_{21}+m_{22} = m$ as:
$$\sqrt{m_{11}}+\sqrt{m_{12}}+\sqrt{m_{21}}+\sqrt{m_{22}} -\lambda(m_{11}+m_{12}+m_{21}+m_{22} - m)$$
and take the partial derivatives with respect to each of the $m_{ij}$'s, setting each to 0. We thus get $\frac12m_{ij}^{-1/2} = \lambda$
for all $i$ and $j$, or $m_{ij} = 1/4\lambda^2$. Since the sum of the four $m_{ij}$'s is $m$, it follows that $m = 1/\lambda^2$, and therefore $m_{ij} = m/4$ at the extremum. Finally, we note that the second derivatives at the extremum are negative, so the extremum is indeed a maximum. We therefore claim that an upper bound on $r_q(n,m)$ occurs when $m_{ij} = m/4$, and this upper bound is
$$r_q(n,m) \le \frac1{2\sqrt{q}}\bigl(\sqrt{m_{11}}+\sqrt{m_{12}}+\sqrt{m_{21}}+\sqrt{m_{22}}\bigr) \le$$
$$\frac1{2\sqrt{q}}4\sqrt{m/4} = \frac1{2\sqrt{q}}2\sqrt{m} = \sqrt{m/q}$$
\end{proof}

\begin{corollary}
\label{alg-c-corr}
$r_q(n,m) \le 2\sqrt{m/q}$ for any $n\ge q/2$.
\end{corollary}

\begin{proof}
If $n\ge q$, we can at most double the number of nodes in each of the sets $X$ and $Y$ and thereby make $n$ be $q$ times a power of 2. By Theorem~\ref{alg-c-th}, the replication rate for this new problem is at most $\sqrt{m/q}$. If we then remove the introduced nodes, we at most double the average replication rate, yielding the corollary.
\end{proof}

As predicted by Theorem~\ref{lower-th}, Algorithm C cannot beat both Algorithms A and B. In fact, note that $\sqrt{m/q}$ is the harmonic mean of the upper bounds $r\le m/n$ and $r\le n/q$ given by Algorithms A and B. That is, $\sqrt{m/q}$ equals $\sqrt{(m/n)(n/q)}$, and therefore cannot be less than both bounds.

\subsection{An Application of Algorithm C}
\label{alg-c-hd-subsect}
While Algorithm C is not useful in the worst case, it can yield a good solution for particular problems. Since the worst case occurs when every partition of the sets $X$ and $Y$ yields a uniform distribution of edges among the four subproblems, we should look for particular partitions that distribute the edges as unevenly as possible. In this section, we explore an example of how this strategy can succeed.

Our problem is a variant of the Hamming-distance-1 problem discussed in Example~\ref{hd-ex}. Let $X$ and $Y$ both be the set of bit strings of length $b$. In this variant, the outputs are pairs of strings $(x,y)$ such that $x$ is in $X$, $y$ is in $Y$, and $y$ is formed from $x$ by changing a single 0 to a 1. That is, we are looking for pairs of strings at Hamming distance 1, but only when the {\em weight} (number of 1's) of the string from $Y$ is one greater than the weight of the string from $X$. Since all bit strings of length $b$ are present in both $X$ and $Y$, we still get all pairs of strings at distance 1.

Our first partition divides both $X$ and $Y$ according to whether the weight of the string is odd or even. That is, let $X_1$ be the set of odd-weight strings in $X$ and $X_2$ be the even weight strings in $X$. Divide $Y$ similarly into $Y_1$ and $Y_2$. Notice that there are no edges connecting nodes in $X_1$ and $Y_1$, and no edges between $X_2$ and $Y_2$. Therefore, only two of the four subproblems, $(X_1,Y_2)$ and $(X_2,Y_1)$ need to be solved. Most importantly, the replication of each input in the full problem is exactly the same as its replication in the one subproblem in which it participates.

We can continue this division recursively. Consider $X_1$, the strings from $X$ of odd weight. These have weights that are either of the form $4i+1$ or $4i+3$, for some integer $i$.
Likewise, the strings in $Y_2$ have weights that are either of the form $4i$ or $4i+2$. If we divide these two sets according to the remainder when the weight is divided by 4, we again create four subproblems, but only two of them have any edges. That is, all edges either connect a string of weight $4i+1$ in $X_1$ to a string of weight $4i+2$ in $Y_2$, or they connect a string of weight $4i+3$ in $X_1$ to a string of weight $4(i+1)$ in $Y_2$. We again find that each input of the subproblem $(X_1,Y_2)$ only participates in one sub-subproblem, so its replication is whatever it is in that sub-subproblem. A similar statement holds for the subproblem $(X_2,Y_1)$. Note that the divisions of sets like $X_1$ into two parts, need not be an even division. However, we still have what we need: the division into subproblems does not create a need for replication of inputs at any level.

Eventually, we run into a limiting factor. After $\log_2b$ divisions of this type, subproblems involve sets of strings of a single weight. The weights in the middle, around $b/2$, have the most strings. The number of strings of weight $b/2$ is the largest, on the order of $2^b/\sqrt{b}$ strings. If $q\ge 2^b/\sqrt{b}$, then at this stage we can solve every subproblem at one reducer. As a result, the replication rate is only 1, as long as $q$ is close to $n=2^b$; precisely, as long as $q\ge 2n/\log n$. That is still significant, since in general we can only get a replication rate as low as 1 if $q=n$. Note that this algorithm is close to the one discussed in \cite{ADSU} that uses weights of strings for efficiency. The differences come from the different model used here (two input sets versus a single input set).

\subsection{Further Decomposition}
\label{further-subsect}
If $q$ is smaller than the size of the largest single-weight set of bit strings, we can find other good ways to partition into subproblems. Let us continue assuming $X$ and $Y$ are sets of bit strings of length $b$, and the problem is to find pairs $(x,y)$ from $X$ and $Y$ respectively, where $y$ is $x$ with a single 0 changed to 1. We shall assume $X$ and $Y$ are all bit strings of length $b$, although they could be only the strings of a single weight and represent a level-$b$ subproblem derived in Section~\ref{alg-c-hd-subsect}.

Another way to decompose $X$ and $Y$ is according to the first bit of the strings. That is, let $X_1$ be the strings in $X$ and begin with 1 and $X_2$ be the strings that begin with 0. Define $Y_1$ and $Y_2$ similarly. First, notice that no edges run between $X_1$ and $Y_0$, so one of the four subproblems can be dropped. Also, the subproblem $(X_2,Y_1)$ has very few edges. In particular, the only edges connect a string $0w$ from $X_2$ with the string $1w$ from $Y_1$. We can therefore use one reducer for each $w$; this reducer receives only the strings $0w$ from $X_2$ and $1w$ from $Y_1$. Its effect is to add 1, for each member of $X_2$ and $Y_1$, to whatever its replication count is from the remaining subproblems.

The other two subproblems are $(X_1,Y_1)$ and $(X_2,Y_2)$. Call these the {\em big} subproblems. All the strings of the first of these begin with 1, and all the strings of the second begin with 0. Therefore, each of these subproblems can be solved exactly as we would the original problem, but with strings of length $b-1$ instead of $b$. Therefore, if we use this decomposition recursively, the replication rate $r(b)$ for strings of length $b$ satisfies the recurrence
$$r(b) = r(b-1) +1/2$$
The justification for this recursion is that half the strings~-- those of $X_2$ and $Y_1$~-- require replication that is 1 for the subproblem $(X_2,Y_1)$ plus $r(b-1)$ for the big subproblem in which it participates. The other half of the strings participate in only one of the big subproblems, and therefore have replication $r(b-1)$.

There is a basis to this recurrence: when $2^b = q$. Then, we need a single reducer to receive all the strings of $X$ and $Y$. That is, $r(\log_2q) = 1$. The solution to the recurrence for reducers of size $q$ is therefore $r(b) = 1 + \frac12(b-\log q)$. Or, using $n=2^b$ as the number of inputs in $X$ and in $Y$, the replication rate is $1+ \frac12\log(n/q)$. This quantity is generally less than the upper bound from \cite{ADSU}, which is $\log n/\log q$, but which only applies when $n$ is at least $q^i$ for some integer $i\ge2$. As a result, this application of Algorithm~C can be superior when $n<q^2$.

\section{Conclusions}
\label{conclusion-sect}
We introduced the class of problems called ``some-pairs,'' for which each output is a function of two inputs, and we considered MapReduce algorithms suitable for solving all problems in this class.  We observed that there are two obvious algorithms for doing so, and depending on the relationship between the numbers of inputs and outputs for the problem, either one could require less communication than the other.  We then showed that, to within a log factor, no general-purpose algorithm can use less communication than the better of the two obvious algorithms.  Finally, we looked at a recursive approach to solving general problems in the some-pairs class that can beat the obvious algorithms for those problems that have an exploitable structure.

\bibliographystyle{abbrv}
\bibliography{semijoin-mapreduce}
\end{document}